\documentclass[conference]{IEEEtran}
\usepackage{graphics,color}
\usepackage{amsmath,amssymb,epsfig,dsfont}
\usepackage{verbatim,subfigure}
\IEEEoverridecommandlockouts

\begin{document}
\title{Collecting Coded Coupons over Generations\\[-3mm]}
\author{
\IEEEauthorblockN{Yao Li}
\IEEEauthorblockA{ECE Department, Rutgers University\\
Piscataway NJ 08854\\
yaoli@winlab.rutgers.edu} \and \IEEEauthorblockN{Emina Soljanin}
\IEEEauthorblockA{Bell Labs, Alcatel-Lucent\\
Murray Hill NJ 07974, USA\\
emina@alcatel-lucent.com} \and \IEEEauthorblockN{Predrag
Spasojevi\'{c}}
\IEEEauthorblockA{ECE Department, Rutgers University\\
Piscataway NJ 08854\\
spasojev@winlab.rutgers.edu}  \thanks{
This work is supported by NSF CNS Grant No. 0721888.}}

\maketitle

\begin{abstract}
To reduce computational complexity and delay in randomized network
coded content distribution (and for some other practical reasons),
coding is not performed simultaneously over all content blocks but
over much smaller subsets known as generations. A penalty is
throughput reduction. We model coding over generations as the coupon
collector's {\em brotherhood} problem. This model enables us to
theoretically compute the expected number of coded packets needed
for successful decoding of the entire content, as well as a bound on
the probability of decoding failure, and further, to quantify the
tradeoff between computational complexity and throughput.
Interestingly, with a moderate increase in the generation size,
throughput quickly approaches link capacity. As an additional
contribution, we derive new results for the generalized collector's
brotherhood problem which can also be used for further study of many
other aspects of coding over generations.

\end{abstract}
\newtheorem{theorem}{Theorem}
\newtheorem{lemma}[theorem]{Lemma}
\newtheorem{claim}[theorem]{Claim}
\newtheorem{corollary}[theorem]{Corollary}
\newtheorem{remark}{Remark}

\section{Introduction}

In P2P systems, such as BitTorrent, content distribution involves
fragmenting the content at its source, and using swarming techniques
to disseminate the fragments among peers. Acquiring a file by
collecting its fragments can be to a certain extent modeled by the
classic coupon collector problem, which indicates some problems such
systems may have. For example, probability of acquiring a novel
fragment drops rapidly as the number of those already collected
increases. In addition, as the number of peers increases, it becomes
harder to do optimal scheduling of distributing fragments to
receivers. One possible solution is to use a heuristic that
prioritizes exchanges of locally rarest fragments. But, when peers
have only local information about the network, such fragments often
fail to match those that are globally rarest. The consequences are,
among others, slower downloads and stalled transfers. Consider file
at a server that consists of $N$ packets, and a client that chooses
a packet from the server uniformly at random with replacement. Then
the process of downloading the file is the classical coupon
collection, and the expected number of samplings needed to acquire
all $N$ coupons is $\mathcal{O}(N\log N)$ (see for example
\cite{feller}), which does not scale well for large $N$ (large
files).

Randomized coding systems, such as Microsoft's Avalanche \cite{avalanche}, attempt to lessen such problems in the following way. Instead of distributing the original
file fragments, peers produce linear combinations of the fragments they already
hold. These combinations are distributed together with a tag that describes the
coefficients used for the combination. When a peer has enough linearly independent
combinations of the original fragments, it can decode and build up the original file. The information packets can be decoded
when $N$ linearly independent equations have been collected. For a
large field, this strategy reduces the expected number of
required samplings from $\mathcal{O}(N\log N)$ to almost $N$ (see for example \cite{monograph}). However, the expense
to pay is the computational complexity. If the information packets
consist of $d$ symbols in $GF(q)$, it takes $\mathcal{O}(Nd)$
operations in $GF(q)$ to form a linear combination per coded packet,
and $\mathcal{O}(N^3+N^2d)$ operations, or, equivalently, $\mathcal{O}(N^2+Nd)$
operations per information packet, to decode the information packets
by solving linear equations.

We refer to the number of information packets combined in a coded
packet as the degree of the coded packet. The complexity  of
computations required  for solving the equations of information
packets depends on the average degree of coded packets. To reduce
computational complexity while maintaining the throughput gain
brought by coding, several approaches have been introduced seeking
to decrease the average degree of coded
packets~\cite{ustc,ltcodes,aminraptor}. Nevertheless, it is hard to
design distributed coding schemes with good throughput/complexity
tradeoff that further combine coded packets.

Chou et al.~\cite{choupractical} proposed to partition information
packets into disjoint \textit{generations}, and combine only packets
from the same generation. The performance of codes with random
scheduling of disjoint generations was first theoretically analyzed
by Maymounkov et al.\ in \cite{petarchunked}, who referred to them
as \emph{chunked codes}. Chunked codes allow convenient encoding at
intermediate nodes, and are readily suitable for peer-to-peer file
dissemination. In \cite{petarchunked}, the authors used an
adversarial schedule as the network model and measured the code
performance by estimating the loss of ``code density'' as packets
are combined at intermediate nodes throughout the network.

In this work, we propose a way to analyze coding with generations
from a coupon collection perspective. Here, we view the generations as
coupons, and model the receiver who needs to acquire multiple linear equations
for each generation as a collector seeking to collect multiple copies of the same coupon. This collecting model is sometimes
refereed to as the \emph{collector's brotherhood} problem, as in \cite{brotherhood}. As a classical probability model which studies random sampling
of a population of distinct elements (coupons) with
replacement \cite{feller}, the \emph{coupon collector's} problem
finds application in a wide range of fields \cite{couponrevisited},
from testing biological cultures for contamination \cite{couponappbio} to probabilistic-packet-marking
(PPM) schemes for IP traceback problem in the Internet \cite{couponapptrace}.
We here use the ideas from the collector's brotherhood problem,
to derive refined results for expected
throughput for finite information length in unicast scenarios.
We describe the tradeoff
in throughput versus computational complexity of coding over
generations. Our results include the asymptotic code performance, by
proportionally increasing either the generation size or the number
of generations. The code performance mean concentration leads us to
a lower bound for the probability of decoding error. Our paper is
organized as follows:  Sec.~\ref{sec:generalmodel} describes the
unicast file distribution with random scheduling and introduces
pertaining results under the coupon collector's model.
Sec.~\ref{sec:throughput} studies the throughput performance of
coding with disjoint generations, including the expected performance
and the probability of decoding failure.  Sec.~\ref{sec:conclusion}
concludes.

\section{Coding over Generations} \label{sec:generalmodel}

\subsection{Coding over Generations in Unicast}

We consider the transmission of a file from a source to a receiver
over a unicast link using network coding over generations.

The file is divided into $N$ information packets, $p_1$, $p_2$,
$\dots$, $p_N$. Each packet $p_i(i=1,2,\dots,N)$ is represented as a
column vector of $d$ information symbols in Galois field $GF(q)$.
Packets are sequentially partitioned into $n=\frac{N}{h}$ {\it
generations} of size $h$, denoted as $G_1$, $G_2$, $\dots$, $G_n$.
$G_j=[p_{(i-1)h+1},p_{(i-1)h+2},\dots,p_{ih}]$ for $j=1,2,\dots,n$.
The receiver collects coded packets from the $n$ generations. The
coding scheme is described as follows:

\paragraph{Encoding}
In each transmission, the source first selects one of the $n$
generations with equal probability. Assume $G_j$ is chosen. Then the
source chooses a coding vector $e$ of length $h$, with each entry
chosen independently and equally probably from $GF(q)$. A new packet
$\bar{p}$ is then formed by linearly combining packets from $G_j$ by
$e$: $\bar{p}=G_je$. The coded packet $\bar{p}$ is then sent over
the communication link to the receiver along with the coding vector
$e$ and the generation index $j$.

\paragraph{Decoding}
The receiver gathers coded packets and their coding vectors. We say
that a receiver collects one more \emph{degree of freedom} for a
generation if it receives a coded packet from that generation with a
coding vector linearly independent of previously received coding
vectors from that generation. Once the receiver has collected $h$
degrees of freedom for a certain generation, it can decode that
generation by solving a system of $h$ linear equations in $h$
unknowns over $GF(q)$.

Note that the two extremes, generation sizes $h=N$ and $h=1$,
correspond respectively to full random linear network coding(i.e.
without generations) and not using coding at all(but with random
piece selection).

Since in the coding scheme described above, both the coding vector
and the generation whose packets are combined are chosen uniformly
at random, the code is inherently rateless. We measure the
throughput by the number of coded packets necessary for decoding all
the information packets.

\paragraph{Computational Complexity} It takes $\mathcal{O}(hd)$ operations
to form each linear combination of $h$ length-$d$ vectors(packets)
in $GF(q)$. The computational cost for encoding is then
$\mathcal{O}(hd)$ per coded packet. Meanwhile, it takes
$\mathcal{O}(h^3+h^2d)$ operations in $GF(q)$ to solve $h$ linearly
independent equations for the $h$ packets of one generation. Thus,
the cost for decoding is $\mathcal{O}(h^2+hd)$ per information
packet.

\subsection{Extension to a General Network}
In a general network, intermediate nodes follow a similar coding
scheme except that previously received coded packets are combined
instead of the original information packets. We leave out the
details since the analysis of code performance in general network
topologies is beyond the scope of this paper.

\subsection{Collector's Problem and Coding Over Generations}
 In a {\it coupon collector's problem}, a set of
distinct coupons are sampled with replacement. Consider the $n$
generations as $n$ distinct coupons, and collecting degrees of
freedom for generation $G_i$ is analogous to collecting copies of
the $i$th element in the coupon set. In the next section, we will
characterize the throughput performance of the code under the coupon
collector's probability model.

\section{Throughput of Coding over Generations} \label{sec:throughput}
We next study the throughput performance of coding
over generations in the simplest single-server-single-user scenario
by using the {\it collector's brotherhood} model \cite{brotherhood}.
Recall that, to successfully recover the file, the receiver has to collect
$h$ degrees of freedom for each generation.
For $i=1,\dots,n$, let $N_i$ be the number of
coded packets sampled from $G_i$ until $h$ degrees of freedom are collected.
Then, $N_i$s are i.i.d. random variables with the expected value \cite{monograph}
\begin{equation} \label{eq:hwait}
E[N_i]=\sum_{j=0}^{h-1}\frac{1}{1-q^{j-h}}.
\end{equation}
Approximating summation by integration, from (\ref{eq:hwait}) we get
\begin{align*}
E[N_i]\lessapprox&\int_0^{h-1}\frac{1}{1-q^{x-h}}dx+\frac{1}{1-q^{-1}}\\
=&h+\frac{q^{-1}}{1-q^{-1}}+\log_q\frac{1-q^{-h}}{1-q^{-1}}.
\end{align*}

Note that $N_i\le s$ means the $h\times s$ matrix formed by $s$
coding vectors of length $h$ as columns is of full row rank. For
$s<h$, $\textnormal{Pr}[N_i\le s]=0$. For $s\ge h$,
$\textnormal{Pr}[N_i\le s]$ equals the probability that for each
$k=1,2,\dots, h,$ the $k$th row in the matrix is linearly
independent of rows $1$ through $(k-1)$. Hence,
\begin{align}
\textnormal{Pr}[N_i\le
s]&=\prod_{k=0}^{h-1}\left(\left(q^s-q^k\right)/q^s\right)
=\prod_{k=0}^{h-1}(1-q^{k-s}).
\end{align}
We have the following Lemma \ref{thm:wait_ccdf} upper bounding the
complementary cumulative distribution function (CCDF) of $N_i$.

\begin{lemma}\label{thm:wait_ccdf}
There exist positive constants $\alpha_{q,h}$ and $\alpha_{2,\infty}$ such that, for $s\ge h$,
\begin{align*}
\textnormal{Pr}[N_i&>s] = 1-\prod_{k=0}^{h-1}(1-q^{k-s})\\
&<1-\exp(-\alpha_{q,h} q^{-(s-h)})<1-\exp(-\alpha_{2,\infty}
q^{-(s-h)}).
\end{align*}
\end{lemma}
\begin{proof}
Please refer to Appendix.
\end{proof}

Let $T(n,m)$ be the number of coded packets collected when for the
first time there are at least $m(\ge1)$ coded packets from every
generation in the collection at the receiver.
The total number of coded packets needed for accumulating $N_i(\ge
h)$ coded packets of each generation $G_i$ is then greater or equal
to $T(n,h)$.

The {\it collector's brotherhood problem}\cite{brotherhood}, also
referred to as the {\it double dixie cup
problem}\cite{doubledixiecup}, investigates the stochastic
quantities associated to acquiring $m(\ge1)$ complete sets of $n$
distinct elements by random sampling.


\subsection{Results From The Collector's Brotherhood Problem}
\label{subsec:probmodel}

For any $m\in\mathbb{N}$, we define $S_m(x)$ as follows:
\begin{align}\label{eq:sm_m}
S_m(x)=&1+\frac{x}{1!}+\frac{x^2}{2!}+\dots+\frac{x^{m-1}}{(m-1)!}\quad(m\ge
1)\\
\label{eq:sm_0}S_{\infty}(x)=& \exp(x) ~\text{and} ~ S_{0}(x)=0.
\end{align}

\begin{theorem}\label{thm:gen_wait}
Consider uniformly random sampling of $n$ distinct coupons with
replacement. Suppose for some $A\in\mathbb{N}$, integers
$k_1,\dots,k_A$ and $m_1,\dots,m_A$ satisfy $1\le
k_1<\dots<k_A\le n$ and $m_1>\dots>m_A\ge1$. For convenience
of notation, let $m_0=\infty$ and $m_{A+1}=0$. Then, the expected
number of samplings needed to acquire at least $m_1$ copies of at
least $k_1$ coupons, at least $m_2$ copies of at least $k_2$
coupons, and so on, at least $m_A$ copies of at least $k_A$ coupons
in the collection is
\begin{align}
&n\int_{0}^{\infty}\Bigl\{e^{nx}-\\
&\sum_{{{(i_0,i_1,\dots,i_{A+1}):\atop i_0=0,i_{A+1}=n}\atop k_j\le
i_j\le i_{j+1}}\atop j=1,2,\dots,A}\!\!
\prod_{j=0}^{A}{{i_{j+1}}\!\!\choose{i_j}}\Bigl[S_{m_{j}}(x)-S_{m_{j+1}}(x)\Bigr]^{i_{j+1}-i_{j}}\Bigr\}e^{-nx}dx\notag
\end{align}
\end{theorem}
\begin{proof}
Our proof generalizes the symbolic method of \cite{doubledixiecup}.
Please refer to Appendix.
\end{proof}

Setting $A=1$, $k_1=k$ and $m_1=m$ in Theorem \ref{thm:gen_wait}
gives the following corollary:
\begin{corollary}\label{thm_waitk}
The expected number of samplings needed to collect at least $k$ of
the $n$ distinct coupons for at least $m$ times is
$n\int_{0}^{\infty}\left\{\sum_{i=0}^{k-1}{n\choose
i}S_{m}^{n-i}(x)\left[e^x-S_{m}(x)\right]^i\right\}e^{-nx}dx$.
\end{corollary}

In the context of coding over generations, this corollary gives us
an estimation of the growth of the size of decodable information. It
is also helpful to the study of a coding scheme with a ``precode''
as discussed in \cite{petarchunked}.


Furthermore, by setting $k=n$ in Corollary \ref{thm_waitk}, we
recover the following result of \cite{doubledixiecup}:
\begin{corollary}
The expected sampling size to acquire $m$ complete sets of $n$ coupons is
\begin{equation}\label{eq:etnm}
E[T(n,m)] = n\int_{0}^{\infty}\left[1-(1-S_m(x)e^{-x})^n\right]dx
\end{equation}
\end{corollary}
(\ref{eq:etnm}) can be numerically evaluated for finite $m$ and $n$.

The asymptotic of $E[T(n,m)]$ for large $n$ has been discussed in
literature such as \cite{doubledixiecup}, \cite{flatto} and
\cite{couponnewaspects}.

\begin{theorem}(\cite{flatto})
When $n\rightarrow\infty$,
\label{eq:tnm_largen}
\[
E[T(n,m)]= n\log n+(m-1)n\log\log n+C_m n +o(n),
\]
 where $C_m=\gamma-\log(m-1)!$, $\gamma$
is Euler's constant and $m\in\mathbb{N}$.
\end{theorem}

For $m\gg1$, on the other hand, we have \cite{doubledixiecup}
\begin{equation}\label{eq:tnm_largem}
E[T(n,m)]\rightarrow nm.
\end{equation}


What is worth mentioning is that, as the number of coupons
$n\rightarrow\infty$, for the first complete set of coupons, the
number of samplings needed is $\mathcal{O}(n\log n)$, while the
additional number of samplings needed for each additional set is
only $\mathcal{O}(n\log\log n)$.

In addition to the expected value of $T(n,m)$, the concentration of
$T(n,m)$ around its mean is also of great interest to us. We can
derive from it an estimate of the probability of successful decoding
after gathering a certain number of coded packets. The generating
function of $T(n,m)$ and its probability distribution are given in
\cite{brotherhood}, but it is quite difficult to evaluate them
numerically. We will instead look at the asymptotic case where the
number of coupons $n\rightarrow \infty$. Erd\"{o}s and R\'{e}nyi
have proven in \cite{renyi} the limit law of $T(n,m)$ as
$n\rightarrow \infty$. Here we restate Lemma B from \cite{flatto} by
Flatto, which in addition expresses the rate of convergence to the
limit law.

\begin{lemma}\cite{flatto}\label{thm:tmn_dist}
Let \[Y(n,m)=\frac{1}{n}\left(T(n,m)-n\log n-(m-1)n\log\log
n\right).\] Then,
\[
\textnormal{Pr}[Y(n,m)\le y] =
\exp\left(-\frac{e^{-y}}{(m-1)!}\right)+\mathcal{O}\left(\frac{\log\log
n}{\log n}\right).
\]
\end{lemma}
\begin{remark}(Remark 2, \cite{flatto})\label{rmk:dist}
The estimation in Lemma \ref{thm:tmn_dist} is understood to hold
uniformly on any finite interval $-a\le y\le a$. i.e., for any
$a>0$,
\[\left|\textnormal{Pr}[Y(n,m)\le y]-\exp\left(-\frac{\exp(-y)}{(m-1)!}\right)\right|\le C(m,a)\frac{\log\log n}{\log n},\]
$n\ge 2$ and $-a\le y\le a$. $C(m,a)$ is a positive constant
depending on $m$ and $a$, but independent of $n$.
\end{remark}



\subsection{Throughput}  \label{subsec:thru_tradeoff}
The number of coded packets $T$ the receiver needs to collect for
successful decoding is lower bounded by $T(n,h)$, and so $E[T]$ is
lower bounded by $E[T(n,h)]$. Also, by Lemma \ref{thm:wait_ccdf},
$N_i$ is well concentrated near $h$ for large $q$, and so $T(n,h)$
could be a good estimate for $T$ for finite $n$. In addition, from
Thm.~\ref{eq:tnm_largen} and (\ref{eq:tnm_largem}) we observe that,
when $n\gg1$ or $m\gg1$, $E[T(n,m)]$ is linear in $m$. Thus, we
could substitute $E[N_i]$ for $m$ in these expressions to roughly
estimate the asymptotic expected number of coded packets needed for
successful decoding.

From Lemma \ref{thm:tmn_dist}, we obtain the following lower
bound to the probability of decoding failure as
$n\rightarrow\infty$:
\begin{theorem}\label{thm:pe_lb}
When $n\rightarrow\infty$, the probability of decoding failure when
$t$ coded symbols have been collected is greater than
$1-\exp\left[-\frac{1}{(h-1)!}n(\log
n)^{h-1}\exp\left(-\frac{t}{n}\right)\right]+\mathcal{O}\left(\frac{\log\log
n}{\log n}\right)$.
\end{theorem}
\begin{proof}
The probability of decoding failure after acquiring $t$ coded
packets equals $\textnormal{Pr}[T>t]$. Since $T\ge T(n,h)$,
\begin{align*}
\textnormal{Pr}[T>t] \ge &\,\textnormal{Pr}[T(n,h)>t]\\
=&1-\!\textnormal{Pr}\left[Y(n,h)\le\frac{t}{n}-\log n-(m-1)\log\log
n\right]
\end{align*}
The result in Theorem \ref{thm:pe_lb} follows directly from Lemma
\ref{thm:tmn_dist}.
\end{proof}

\begin{figure}[htbp]
\begin{center}
\subfigure[]{\label{subfig:T}\includegraphics[scale=0.5]{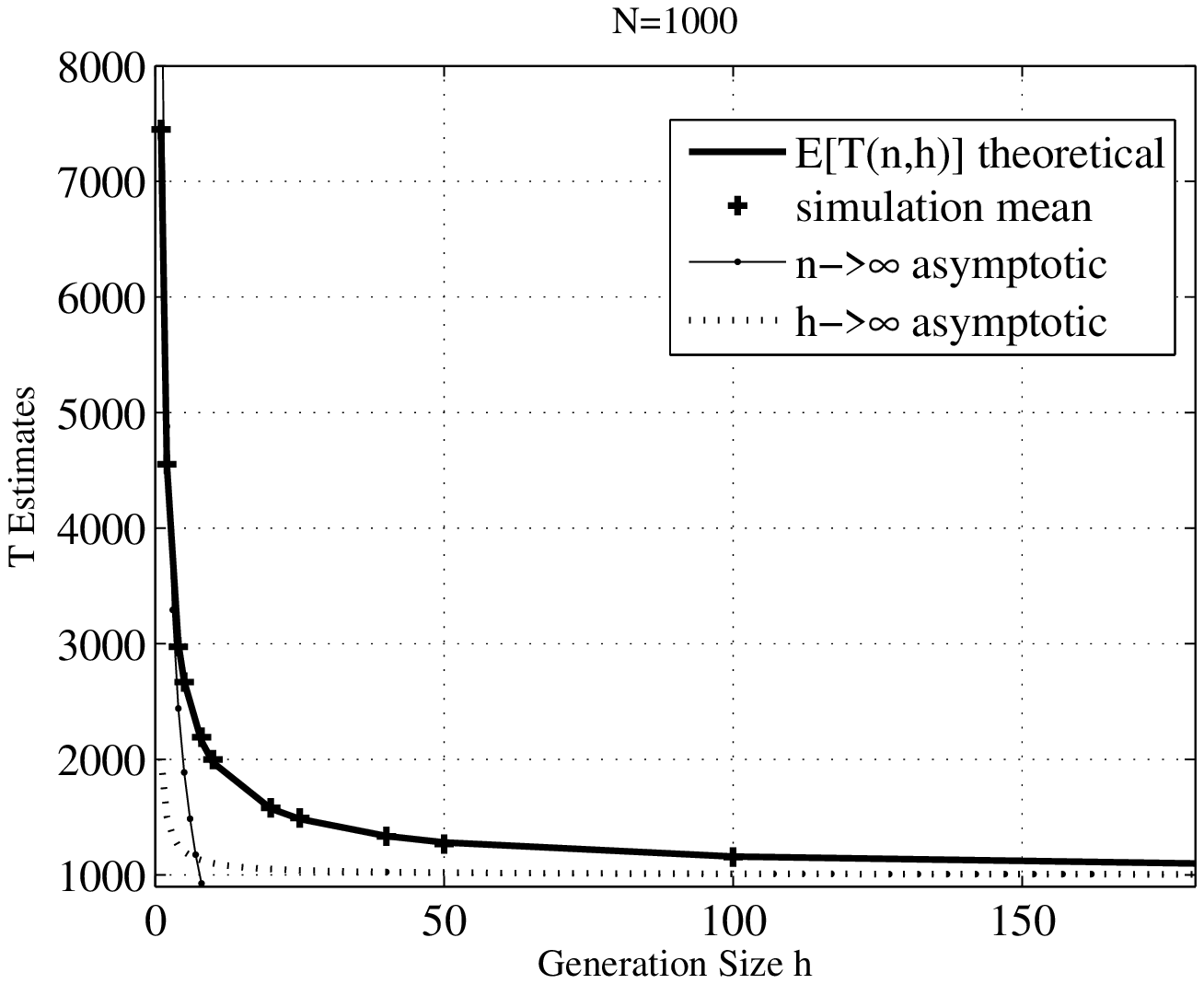}}\qquad
\subfigure[]{\label{subfig:pe}\includegraphics[scale=0.5]{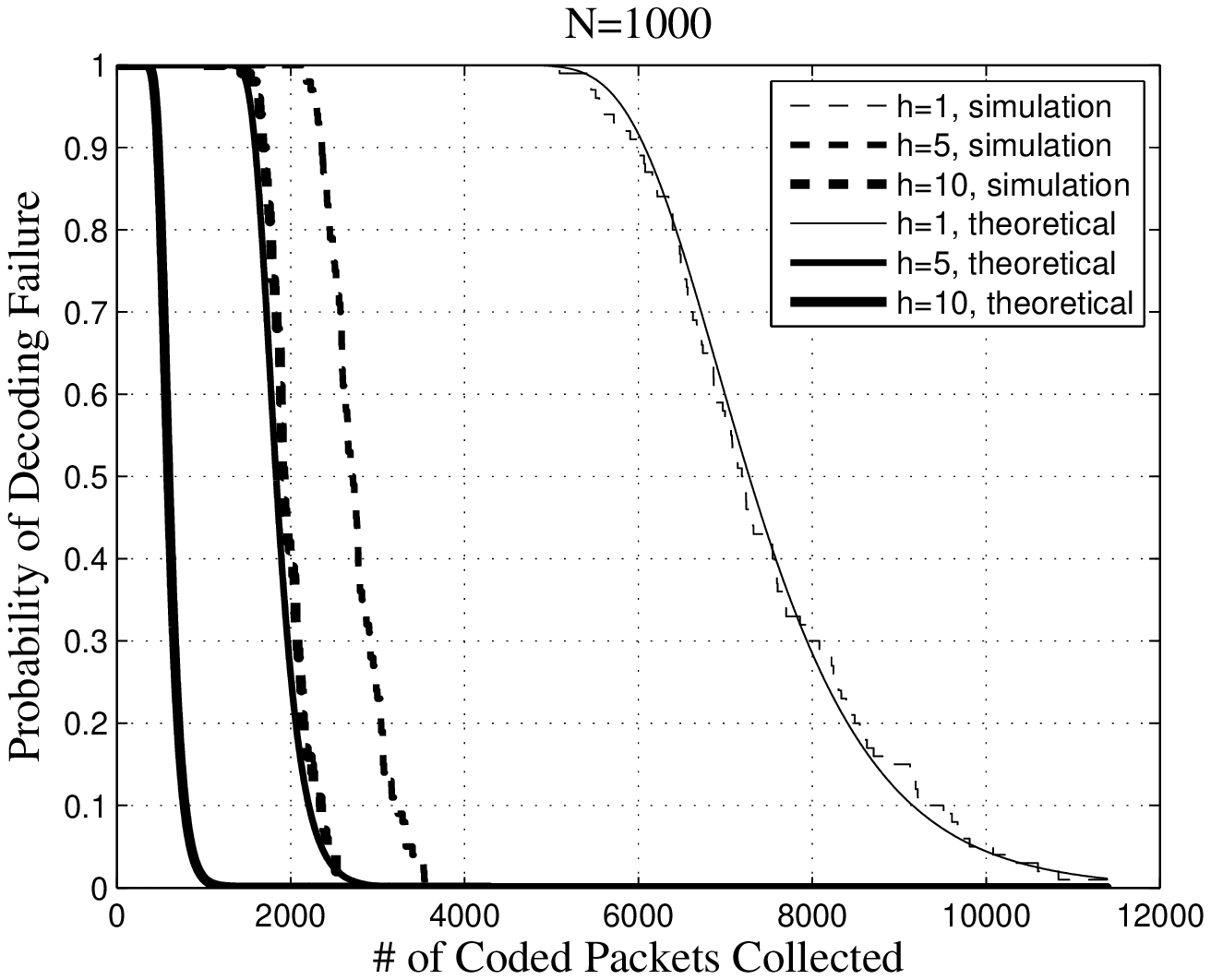}}
\caption{(a)Estimates of $T$, the number of coded packets needed for
successful decoding when the total number of information packets
$N=1000$: $E[T(n,h)]$; average of $T$ in simulation ($q=256$);
$n\rightarrow\infty$ asymptotic Theorem~\ref{eq:tnm_largen}; $m\gg1$
asymptotics (\ref{eq:tnm_largem}).
(b)Estimates of probability of decoding failure versus the number of
coded packets collected: Theorem \ref{thm:pe_lb} along with
simulation results ($q=256$).} \label{fig:brotherhood}
\end{center}
\end{figure}

Figure \ref{fig:brotherhood}\subref{subfig:T} shows several
estimates of $T$ for fixed $N=nh=1000$ versus generation size
$h$. 
It is worth noting that, for fixed $N$, $E[T(n,h)]$ drops
significantly as $h$ is increased to a relatively small value. For
the example $N=1000$, when each generation includes $1/10$ of all
the information packets, the expected overhead required for
successful decoding is below $16\%$. In such scenarios as assumed in
our model, where feedback is minimal, the benefit of coding on
throughput is pronounced even when it is done in relatively small
information block length. In effect, communication overhead due to
the exchange of control messages makes way for moderate
computational complexity at individual nodes. This is particularly
meaningful to communication networks in shared media, in which
contention occurs frequently, and also to networks of nodes with
medium computing power.

Figure \ref{fig:brotherhood}\subref{subfig:pe} shows the estimate of
the probability of decoding failure versus $T$. As pointed out in
Remark \ref{rmk:dist}, the deviation of the CDF of $T(n,m)$ from the
limit law for $n\rightarrow\infty$ depends on $m$ and is on the
order of $\mathcal{O}({\frac{\log\log n}{\log n}})$, which is quite
slow. This is also implied in our observation from Figure
\ref{fig:brotherhood}\subref{subfig:T} that Thm.~\ref{eq:tnm_largen}
gives a good estimate of $E[T(n,m)]$ only for very small values of
$m$(compared to $n$).

\section{Conclusion And Future Work}\label{sec:conclusion}
We investigated the throughput performance of coding over
generations in the unicast scenario under the classical yet
ever-useful coupon collector's model. We derived a general formula
(Theorem \ref{thm:gen_wait}) to compute the expected number of
samplings necessary for collecting multiple copies of $n$ distinct
coupons. The formula can be applied in various ways to analyze a
number of different aspects of coding over generations. Here in
particular, we used a special case of this result, namely, the
expected number of samplings $E[T(n,h)]$ needed to collect $h$
complete sets of $n$ distinct coupons to estimate the expected
number of coded packets necessary for successful decoding when $N$
information packets are encoded over $n$ generations of size $h$
each. Apart from results for finite information length $N$, the
asymptotics of $E[T(n,h)]$ can also be used to estimate the
performance of coding when either the number of generations $n$ or
the generation size $h$ go to inifinity. We also gave a lower bound
for the probability of decoding failure by using the limit law for
$T(n,h)$ as $n\rightarrow\infty$.

The general result expressed in Theorem \ref{thm:gen_wait} has
proved its usefulness to the analysis of coding over overlapping
generations in our recent work~\cite{overlapping}. Further, we have
been able to derive in~\cite{coupon_j} the exact expression, as well
as an upper bound, for the expected number of coded packets
necessary for successful decoding. We expect to extend our analysis
under the coupon collection framework to many other aspects of
coding over generations.

\bibliographystyle{IEEEtran}
\bibliography{netcodrefs}
\appendix
\subsection*{Proof of Lemma~\ref{thm:wait_ccdf}}
\noindent For $i=1,2,\dots,n$ and any $s\ge h$, we have
\begin{align*}
&\ln\textnormal{Pr}\bigl\{N_i\le s\bigr\}
 =\sum_{k=0}^{h-1}\ln(1-q^{k-s})=-\sum_{k=0}^{h-1}\sum_{j=1}^{\infty}\frac{1}{j}q^{(k-s)j}\\
 =&-\sum_{j=1}^{\infty}\frac{1}{j}\sum_{k=0}^{h-1}q^{j(k-s)}
 =-\sum_{j=1}^{\infty}\frac{1}{j}q^{-js}\frac{q^{jh}-1}{q^{j}-1}\\
 =&-q^{-(s-h)}\sum_{j=1}^{\infty}\frac{1}{j}q^{-(j-1)(s-h)}\frac{1-q^{-jh}}{q^{j}-1}\\
 >&q^{-(s-h)}\sum_{j=1}^{\infty}\frac{1}{j}\frac{1-q^{-jh}}{1-q^{j}} =q^{-(s-h)}\ln \textnormal{Pr}\bigl\{N_i\le h\bigr\}\\
 >&q^{-(s-h)}\lim_{h\rightarrow\infty,q=2}\ln \textnormal{Pr}\bigl\{N_i\le h\bigr\}
\end{align*}
The claim is obtained by setting \[
\alpha_{q,h}=-\ln \textnormal{Pr}\bigl\{N_i\le h\bigr\},\; \alpha_{2,\infty}=-\!\!\lim_{h\rightarrow\infty,q=2}\ln
\textnormal{Pr}\bigl\{N_i\le h\bigr\}.
\]

\subsection*{Proof of Theorem~\ref{thm:gen_wait}}
\noindent Our proof generalizes the symbolic method of
\cite{doubledixiecup}. Let $E$ be the event that, in the acquired
collection, there are at least $m_1$ copies of at least $k_1$
coupons, at least $m_2$ copies of at least $k_2$ coupons, and so on,
at least $m_A$ copies of at least $k_A$ coupons. For $t\ge 0$, let
$E(t)$ be the event that $E$ has occurred after $t$ samplings, and
let $\mathds{1}_{\bar{E}(t)}$ be the indicator that takes value $1$
if $E(t)$ does not occur and $0$ otherwise. Then random variable
$W=\mathds{1}_{\bar{E}(0)}+\mathds{1}_{\bar{E}(1)}+\dots$ equals the
waiting time for $E$ to occur. For $t\ge 0$, let
$\pi_t=\textnormal{Prob}[\bar{E}(t)]$; we then have
\begin{equation}\label{eq:waitsum}
E[W]=\sum_{t\ge 0}\pi_t.
\end{equation}

To derive $\pi_t$, we introduce an operator $f$ acting on an
$n$-variable polynomial $g$. For a monomial $x_1^{w_1}\dots
x_n^{w_n}$, let $i_j$ be the number of exponents $w_u$ among
$w_1,\dots,w_n$ satisfying $w_u\ge m_j$, for $j=1,\dots,A$. $f$
removes all monomials $x_1^{w_1}\dots x_n^{w_n}$ in $g$ satisfying
$i_1\ge k_1, \dots,i_A\ge k_A$ and $i_1\le \dots\le i_A$. Note that
$f$ is a linear operator, i.e., if $g_1$ and $g_2$ are two
polynomials in the same $n$-variables, and $a$ and $b$ two scalars,
we have $af(g_1)+bf(g_2)=f(ag_1+bg_2)$.

Each monomial in $(x_1+\dots+x_n)^t$ corresponds to one of the $n^t$ possible outcomes
after $t$ samplings, with the exponent of $x_i$ being the number of collected
copies of the $i$th coupon. Hence, the number of outcomes counted in
$\bar{E}(t)$ equals $f((x_1+\dots+x_n)^t)$ evaluated at
$x_1=\dots=x_n=1$.
\[\pi_t=\frac{f((x_1+\dots+x_n)^t)}{n^t}|_{x_1=\dots=x_n=1},\]
and thus, by (\ref{eq:waitsum}), $
E[W]=\sum_{t\ge0}\frac{f\left((x_1+\dots+x_n)^t\right)}{n^t}|_{x_1=\dots=x_n=1}.
$ Making use of the identity
$\frac{1}{n^t}=n\int_{0}^{\infty}\frac{1}{t!}y^te^{-ny}dy$, and
because of the linearity of operator $f$, we further have
\begin{align}
E[W]
=&n\int_{0}^{\infty}\sum_{t\ge0}\frac{f\left((x_1+\dots+x_n)^t\right)}{t!}y^te^{-ny}dy \notag\\
=&n\int_{0}^{\infty}f\Bigl(\sum_{t\ge0}\frac{(x_1y+\dots+x_ny)^t}{t!}\Bigr)e^{-ny}dy \notag\\
=&n\int_{0}^{\infty}f\left(\exp(x_1y+\dots+x_ny)\right)
e^{-ny}dy \label{eq:expsum}
\end{align}
evaluated at $x_1=\dots=x_n=1$.

We next need to find the sum of the monomials in the polynomial
expansion of $\exp(x_1+\dots+x_n)$ that should be removed under $f$.
If we choose integers $0=i_0\le  i_1\le\dots\le i_A\le
i_{A+1}=n$, such that $i_{j}\ge k_{j}$ for $j=1,\dots,A$, and
then partition indices $\{1,\dots,n\}$ into $(A+1)$ subsets
$\mathcal{I}_1,\dots,\mathcal{I}_{A+1}$, where
$\mathcal{I}_j(j=1,\dots,A+1)$ has $i_{j}-i_{j-1}$ elements. Then
\begin{equation}\label{eq:expprodform}
\prod_{j=1}^{A+1}\prod_{i\in\mathcal{I}_{j}}(S_{m_{j-1}}(x_i)-S_{m_j}(x_i))
\end{equation}
equals the sum of all monomials in $\exp(x_1+\dots+x_n)$ with
($i_{j}-i_{j-1}$) of the $n$ exponents smaller than $m_{j-1}$ but
greater than or equal to $m_{j}$, for $j=1,\dots,A+1$. (Here $S$ is
as defined by (\ref{eq:sm_m}-{\ref{eq:sm_0}}).) The number of such
partitions of $\{1,\dots,n\}$ is equal to
${{n}\choose{n-i_A,\dots,i_2-i_1,i_1}}=\prod_{j=0}^{A}{{i_{j+1}}\choose
{i_{j}}}$. Finally, we need to sum the terms of the form
(\ref{eq:expprodform}) over all partitions of all choices of
$i_1,\dots,i_A$ satisfying $k_{j}\le i_{j}\le i_{j+1}$ for
$j=1,\dots,A$:
\begin{align}\label{eq:fform}
&f\left(\exp(x_1y+\dots+x_ny)\right)|_{x_1=\dots=x_n=1}
=\exp(ny)-\notag\\
&\sum_{{{(i_0,i_1,\dots,i_{A+1}):\atop i_0=0,i_{A+1}=n}\atop k_j\le
i_j\le i_{j+1}}\atop
j=1,2,\dots,A}\prod_{j=0}^{A}{{i_{j+1}}\choose{i_j}}\left[S_{m_{j}}(y)-S_{m_{j+1}}(y)\right]^{i_{j+1}-i_{j}}.
\end{align}
Bringing (\ref{eq:fform}) into (\ref{eq:expsum}) gives our result in
Theorem \ref{thm:gen_wait}.

\end{document}